\documentclass{article}  

\usepackage{arxiv}
\usepackage{amsthm}
\usepackage[utf8]{inputenc}
\usepackage{amsmath}
\usepackage{amsfonts}
\usepackage{amssymb}
\usepackage[all]{xy}

\usepackage{xcolor}  

\usepackage{amssymb}

\usepackage{amsmath}
\usepackage{mathdots}

\usepackage{mathtools}

\usepackage{tensor}

\usepackage{urwchancal}
\DeclareFontFamily{OT1}{pzc}{}
\DeclareFontShape{OT1}{pzc}{m}{it}{<-> s * [1.10] pzcmi7t}{}
\DeclareMathAlphabet{\mathpzc}{OT1}{pzc}{m}{it}

\usepackage{graphicx}
\usepackage{ifpdf}
\ifpdf
\usepackage{epstopdf}
\PrependGraphicsExtensions{.svg}
\fi

\newtheorem{theorem}{Theorem}[section]
\newtheorem{lemma}[theorem]{Lemma}

\newtheorem{proposition}[theorem]{Proposition}

\providecommand{\R}{\mathbb{R}}

\providecommand{\SO}{\mathbf{SO}}

\providecommand{\MR}{\mathbf{MR}}

\providecommand{\grpG}{\mathbf{G}}

\providecommand{\gothg}{\mathfrak{g}}

\providecommand{\gothX}{\mathfrak{X}} 

\providecommand{\Sph}{\mathrm{S}}

\providecommand{\calM}{\mathcal{M}}
\providecommand{\calN}{\mathcal{N}}

\providecommand{\calV}{\mathcal{V}}
\providecommand{\calV}{\mathcal{W}}

\providecommand{\vecV}{\mathbb{V}}

\providecommand{\calV}{\mathcal{V}}

\providecommand{\Sym}{\mathbb{S}} 

\providecommand{\tT}{\mathrm{T}} 

\providecommand{\calf}{\mathpzc{f}} 

\DeclareMathOperator{\diag}{diag}
\DeclareMathOperator{\stab}{stab}
\DeclareMathOperator{\Ad}{Ad}

\DeclareMathOperator{\image}{im}

\providecommand{\id}{\mathrm{id}} 

\providecommand{\td}{\mathrm{d}}
\providecommand{\tD}{\mathrm{D}}

\providecommand{\ddt}{\frac{\td}{\td t}}

\providecommand{\mr}[1]{{{#1}^\circ}} 

\usepackage{accents}
\makeatletter
\providecommand{\scirc}{%
    \hbox{\fontfamily{\rmdefault}\fontsize{0.4\dimexpr(\f@size pt)}{0}\selectfont{\raisebox{-0.52ex}[0ex][-0.52ex]{$\circ$}}}}

\makeatother

\mathchardef\mhyphen="2D

\providecommand{\etal}{\textit{et al.}~}

\renewcommand{\mr}[1]{#1^\circ}

\usepackage{hyperref}

\providecommand{\tT}{\mathrm{T}}

\begin{document}






\newcommand{\publicationdetails}
{
  \copyrightNoticeIEEE{2019}
  P. v. Goor, R. Mahony and T. Hamel, "Equivariant Filter (EqF): A General Filter Design for Systems on Homogeneous Spaces," \textit{2020 IEEE 59th Conference on Decision and Control (CDC)}, 2020, pp. 5401-5408, IEEE
  \DOILink{https://ieeexplore.ieee.org/abstract/document/9303813}{10.1109/CDC42340.2020.9303813}
}
\newcommand{\publicationversion}
{Author accepted version}

\title{Equivariant Filter (EqF): A General Filter Design for Systems on Homogeneous Spaces}
\headertitle{Equivariant Filter (EqF): A General Filter Design for Systems on Homogeneous Spaces}

\author{
\href{https://orcid.org/0000-0003-4391-7014}{\includegraphics[scale=0.06]{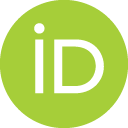}\hspace{1mm}
Pieter van Goor} \\
Department of Electrical, Energy and Materials Engineering\\
Australian National University\\
  ACT, 2601, Australia \\
\texttt{Pieter.vanGoor@anu.edu.au}
\\ \And	
\href{https://orcid.org/0000-0002-7779-1264}{\includegraphics[scale=0.06]{orcid.png}\hspace{1mm}
Tarek Hamel} \\
I3S (University C\^ote d'Azur, CNRS, Sophia Antipolis)\\
and Insitut Universitaire de France\\
\texttt{THamel@i3s.unice.fr}
\\ \And 
\href{https://orcid.org/0000-0002-7803-2868}{\includegraphics[scale=0.06]{orcid.png}\hspace{1mm}
Robert Mahony} \\
Department of Electrical, Energy and Materials Engineering\\
Australian National University\\
  ACT, 2601, Australia \\
\texttt{Robert.Mahony@anu.edu.au}
}

\maketitle

\begin{abstract}
The kinematics of many mechanical systems encountered in robotics and other fields,
such as single-bearing attitude estimation and SLAM,
are naturally posed on homogeneous spaces: That is, their state lies in a smooth manifold equipped with a transitive Lie-group symmetry.
This paper shows that any system posed in a homogeneous space can be extended to a larger system that is equivariant under a symmetry action.
The equivariant structure of the system is exploited to propose a novel new filter, the Equivariant Filter (EqF), based on linearisation of global error dynamics derived from the symmetry action.
The EqF is applied to an example of estimating the positions of stationary landmarks relative to a moving monocular camera that is intractable for previously proposed symmetry based filter design methodologies.
\end{abstract}


\section{Introduction}


The importance of Lie group symmetries in analysing non-linear systems has been recognised since the 1970s \cite{1972_brocket_SIAM,1972_jurdjevic_JDE,1973_brocket_SIAM}.
Jurdjevic and Sussmann generalised the ideas of Brockett \cite{1972_brocket_SIAM} to abstract Lie groups \cite{1972_jurdjevic_JDE}.
Cheng \etal obtained necessary and sufficient conditions for observability
\cite{1990_cheng_SIAM}.
A comprehensive discussion of these early results can be found in Chapter 6 of Jurdjevic's book \cite{1997_jurdjevic_geometric_control}.

In one of the earliest works applying Lie group symmetry for observer design for an explicit example, Salcudean \cite{1991_salcudean_TAC} proposed a non-linear observer for attitude estimation of a satellite using the quaternion representation of rotation.
Thienel \etal  \cite{2003_thienel_TAC} added an analysis of observability and  bias estimation.
In parallel work, Aghannan proposed a general observer design methodology for Lagrangian systems by exploiting invariance properties \cite{2003_aghannan_TAC}.
Driven by the emerging aerial robotics community and the need for robust and simple attitude observers Mahony \etal \cite{2008_Mahony_tac} developed a non-linear observer for rotation matrices posed directly on the matrix Lie group $\SO(3)$ with almost-global convergence properties.
In parallel, Bonnabel \etal \cite{2007_Bonnabel_cdc} proposed the left-invariant extended Kalman filter and applied this to attitude estimation.
Both these observers are fundamentally derived from the symmetry properties of the underlying system and had significant impact in the robotics community.
This motivated more general studies of systems on Lie groups and homogeneous spaces \cite{2008_Bonnabel_TAC,2009_Bonnabel_TAC,2009_lageman_TAC}.
In \cite{2009_Bonnabel_cdc} Bonnabel \etal proposed the Invariant Extended Kalman Filter (IEKF), a sophisticated observer design for systems on Lie groups with invariance properties.
Mahony \etal considered observer design for equivariant kinematic systems on homogeneous spaces with equivariant output functions using Lyapunov design principles \cite{RM_2013_Mahony_nolcos}.
This general design was extended in \cite{2015_khosravian_automatica} to deal with biased input measurements, and to use general gain mappings.
Work by Barrau and Bonnabel \cite{2017_Barrau_tac} extended the IEKF from invariant systems to a broader class of ``group affine'' systems, and characterised the filter's convergence properties.
Recent symmetry-based observers have demonstrated how equivariant designs naturally handle invariance \cite{2016_Barrau_arxive,2017_Mahony_cdc} in a class of systems for the classical Simultaneous Localisation and Mapping (SLAM) problem in robotics.
This paper draws from recent work on observer design for systems on Lie groups \cite{2020_mahony_mtns,2020_mahony_annrev}.



In this paper, we study equivariance properties of kinematic systems on homogeneous spaces and propose a novel observer design methodology for systems with a Lie group state symmetry.
One of our key contributions is that we show that any system posed on a homogeneous space can be embedded in a natural manner into an extended system that is equivariant.
This allows the proposed observer design to be applied to any system on a homogeneous space by extending the system if necessary.
We pose the observer as a lift of the extended equivariant system on the symmetry group associated with the homogeneous space.
We go on to propose a novel observer design, the Equivariant Filter (EqF), based on linearising the dynamics of global error coordinates derived from the system symmetry and lifting the associated gain map back to the observer state system.
The approach differs from an extended Kalman filter approach by introducing global error coordinates (derived from the symmetry) and linearising the error dynamics around a fixed equilibrium, rather than linearising along the observer trajectory.
All previous (equivariant systems) observer/filter design methodologies depend on assuming properties of the system; invariance for the constructive designs \cite{RM_2013_Mahony_nolcos} or group affine structure for the IEKF designs \cite{2017_Barrau_tac}.
The EqF requires only that the system state is a homogeneous manifold with associated Lie-group symmetry.
It specialises to the Invariant Extended Kalman Filter (IEKF) \cite{2017_Barrau_tac} when the lifted system displays the specific ``group affine'' structure required for the IEKF derivation.
Unlike the IEKF, the EqF can also be applied to systems posed on homogeneous spaces rather than on Lie groups.
Moreover, it enables the user to choose local coordinates to modify the filter properties, and it fully exploits the symmetry properties of the system without requiring that the system model is posed explicitly on the Lie-group, a process that would lead to unobservable states associated with the stabiliser subgroups.
We apply the proposed filter to an example problem, 2D derotated ego-centric visual SLAM, of determining the positions of 2D landmarks with respect to a moving camera from bearing (or visual) measurements.
This example is of interest since it does not have invariant or group affine system structure and prior design methodologies do not apply.
By extending the velocity space we generate an equivariant system which is still not invariant nor group affine, but to which the proposed EqF can be applied.
The simulations demonstrate the performance of the proposed design methodology.




\section{Preliminaries}


For a smooth manifold $\calM$, let $\tT_\xi \calM$ denote the tangent space of $\calM$ at $\xi$, let $\tT\calM$ denote the tangent bundle, and let $\gothX(\calM)$ denote the space of vector fields over $\calM$.
Given a vector field $f \in \gothX(\calM)$, $f (\xi) \in \tT_\xi \calM$ denotes the value of the vector field $f$ at $\xi \in \calM$.

Given a differentiable function between smooth manifolds $h:\calM \to \calN$, the map
\begin{align*}
    \tD h(\xi): \tT_\xi \calM &\to  \tT_{h(\xi)} \calN, \\
    v &\mapsto \tD h(\xi)[v],
\end{align*}
denotes the derivative of $h$ at $\xi$.
Additionally, the map
\begin{align*}
    \td h: \tT \calM &\to \tT \calN, \\
    (\xi,v) &\mapsto (h(\xi), \tD h(\xi)[v]),
\end{align*}
denotes the differential of $h$ where the base point is implicit in the argument.
Given a function $\Lambda : \calM_1\times \cdots \times \calM_n \to \calN$, the map $\tD_{\xi^i} \Lambda : \tT_{\xi^i} \calM_i \to \tT \calN$ is defined as the partial differential of $\Lambda$ with respect to the argument $\xi^i$.



A general Lie group is denoted $\grpG$ and has Lie algebra $\gothg$.
The identity element is written $\id \in \grpG$.
For any $X \in \grpG$, the left and right translations by $X$ are denoted $L_X$ and $R_X$, respectively, and are defined by
\begin{align*}
    L_X(Y) := XY, \qquad R_X(Y) := YX
\end{align*}
where $XY$ denotes the group multiplication.
The adjoint map $\Ad: \grpG \times \gothg \to \gothg$ is defined by
\begin{align*}
    \Ad_X U = \td L_X \td R_{X^{-1}} U
\end{align*}
for every $X \in \grpG$ and $U \in \gothg$.

A right action of a Lie-group $\grpG$ on a manifold $\calM$ is a smooth map $\phi: \grpG \times \calM \to \calM$ that satisfies,
\begin{align*}
    \phi(Y, \phi(X, \xi)) = \phi(XY, \xi),
\end{align*}
for any $X,Y \in \grpG$ and any $\xi \in \calM$.
For a fixed $X \in \grpG$, the partial map $\phi_X : \calM \to \calM$ is defined by
$\phi_X(\xi) := \phi(X,\xi)$.
Likewise, for a fixed $\xi \in \calM$, the partial map $\phi_\xi : \grpG \to \calM$ is defined by
$\phi_\xi(X) := \phi(X,\xi)$.
We will assume that all symmetries in the paper are right actions, noting that any left action can be transformed to a right action by considering the inverse parameterization of the group.

\begin{proposition}
Any right action $\phi: \grpG \times \calM \to \calM$ induces a right action on vector fields over $\calM$, denoted $\td_\star \phi: \grpG \times \gothX(\calM) \to \gothX(\calM)$, and defined by
\begin{align*}
    \td_\star \phi (X, f) := \td \phi_X f \circ \phi_X^{-1},
\end{align*}
for any $f \in \gothX(\calM)$ and $X \in \grpG$.
For a fixed $X \in \grpG$, $\td_\star \phi_X$ is a linear map on $\gothX(\calM)$.
\end{proposition}

\begin{proof}
To see that $\td_\star \phi$ is a group action, let $X,Y \in \grpG$ and $f \in \gothX(\calM)$.
Then,
\begin{align*}
    \td_\star \phi(Y, \td_\star \phi(X, f))
    &= \td \phi_Y \td \phi_X f \circ \phi_X^{-1} \circ \phi_Y^{-1}, \\
    &= \td \phi_{XY} f \circ \phi_{Y^{-1}X^{-1}}, \\
    &= \td_\star \phi(XY, f).
\end{align*}
To prove linearity
let $c_1,c_2 \in \R$ and let $f_1, f_2 \in \gothX(\calM)$, then
\begin{align*}
    \td_\star & \phi(X, c_1 f_1 + c_2 f_2)(\xi) \\
    &= \td \phi_X (c_1 f_1 + c_2 f_2)({\phi_X^{-1}(\xi)}), \\
    &= c_1 \td \phi_X (f_1)({\phi_X^{-1}(\xi)}) + c_2 \td \phi_X (f_2)({\phi_X^{-1}(\xi)}), \\
    &= c_1 \td_\star \phi(X, f_1)(\xi) + c_2 \td_\star \phi(X, f_2)(\xi),
\end{align*}
where the second-last line follows from the linearity of the differential $\td \phi_X$.
\end{proof}


\section{Problem Description}


Let $\calM$ be a smooth $M-dimensional$ manifold termed the \textit{state space}.
A \emph{system function} for a kinematic system on $\calM$ is a linear map
\begin{align} \label{eq:kinematic_system}
    f: \vecV &\to \gothX(\calM), \notag \\
    v &\mapsto f_v,
\end{align}
where $\vecV$ is a real vector space termed the \textit{velocity space}.
Trajectories $\xi(t) \in \calM$ on a time interval $[0,\infty)$ of the system considered are solutions of the ordinary differential equation
\begin{align}
    \dot{\xi} &= f_{v(t)}(\xi), \quad\quad\quad \xi(0) \in \calM, \label{eq:kin_system_trajectory}
\end{align}
with initial condition $\xi(0)$ and measured input signal $v(t) \in \vecV$.
We will assume $v(t)$ is sufficiently smooth to ensure unique well defined solutions for all time.
The system output function for a kinematic system is a function
\begin{align} \label{eq:measurement_system}
    h:\calM \to \calN,
\end{align}
where $\calN$ is a smooth manifold termed the \textit{output space}.
In this paper, we will assume that the output space $\calN$ is embedded in $\R^N$ for some sufficiently large $N$ to reduce the complexity of the analysis.
We will use the embedding coordinates $ h(\xi_1), h(\xi_2)  \in \calN \hookrightarrow \R^N$ explicitly in computing output errors $h(\xi_1) - h(\xi_2) \in \R^N \not\in \calN$ in the proposed filter update \eqref{eq:eqf_delta_innovation}.

Let $\grpG$ be a Lie group with Lie algebra $\gothg$, and suppose that $\calM$ is a homogeneous space of $\grpG$; that is, there exists a smooth transitive right group action of $\grpG$ on $\calM$,
\begin{align} \label{eq:phi_action}
    \phi: \grpG \times \calM \to \calM.
\end{align}
A \textit{lift} for the system \eqref{eq:kinematic_system} is a function $\Lambda: \calM \times \vecV \to \gothg$ satisfying
\begin{align} \label{eq:lift_condition}
    \tD \phi_\xi(\id)\left[ \Lambda(\xi, v) \right] &= f_v(\xi),
\end{align}
for every $\xi \in \calM$ and $v \in \vecV$ \cite{RM_2013_Mahony_nolcos}.



Any kinematic system \eqref{eq:kinematic_system} defined on a homogeneous space admits a lift $\Lambda: \calM \times \vecV \to \gothg$ satisfying \eqref{eq:lift_condition} \cite{RM_2013_Mahony_nolcos}.
In the particular case where $\tD \phi_\xi(\id)$ is invertible, corresponding to a group action where $\stab_\phi (\xi)$ is trivial or discrete, the lift is unique \cite{2020_mahony_mtns}.


\section{Equivariant Systems}


\subsection{Kinematic Equivariance}
\label{sec:kinematic_equivariance}


Consider the system function $f : \vecV \to \gothX(\calM)$.
We will assume that $\ker f = \{0\}$ is trivial, that is there is no $v \in \vecV$ such that $f_v(\xi) \equiv 0$ for all $\xi \in \calM$.
To see there is no loss of generality in this assumption, assume that $\ker f \neq \{0\}$.
Consider a new input vector space $\vecV' = \vecV / \ker f$ and a new system function $f'  : \vecV' \to \gothX(\calM)$ with $f'_{\{u + \ker f\}} = f_{u}$.
It is easily verified that $f'$ is well defined, linear in the new input variable, and has trivial kernel.

A kinematic system \eqref{eq:kinematic_system} is termed \textit{equivariant} if there exists a right group action $\psi: \grpG \times \vecV \to \vecV$ such that
\begin{align*}
    \td \phi_X f_v(\xi) &= f_{\psi_X(v)}(\phi_X(\xi)),
\end{align*}
for all $v \in \vecV$, $\xi \in \calM$ and $X \in \grpG$.
In this case, one has
\begin{align}
    f_{\psi_X(v)}(\xi)
    &= f_{\psi_X(v)}(\phi_X(\phi_X^{-1}(\xi))), \notag \\
    &= \td \phi_X f_v(\phi_X^{-1}(\xi)), \notag \\
    &= \td_\star \phi_X f_v(\xi), \label{eq:equivalent_def_equivariance}
\end{align}
That is $f_{\psi_X(v)} = \td_\star \phi_X f_v$ as elements of $\gothX(\calM)$.
Equivalently, the diagram
\begin{equation*}
    \xymatrix{
        \vecV \ar[r]^{f} \ar[d]_{\psi_X} & \gothX(\calM) \ar[d]^{\td_\star \phi_X} \\
        \vecV \ar[r]^{f} & \gothX(\calM) \\
    }
\end{equation*}
commutes for all $X \in \grpG$.

From \eqref{eq:equivalent_def_equivariance} it follows that if an action $\psi$ exists for which $f$ is equivariant, then it is fully determined by $\td_\star \phi$.
However, the linear action $\td_\star \phi$ is defined on all $\gothX(\calM)$ and the right hand side of \eqref{eq:equivalent_def_equivariance} is always defined.
In particular, if $f : \vecV \to \gothX(\calM)$ is a kinematic system and $\td_\star \phi : \image(f) \to \image(f)$ then the right hand side of \eqref{eq:equivalent_def_equivariance} can be used to define the group action $\psi : \grpG \times \vecV \to \vecV$ by
\[
\psi_X(v) := f^{-1}(\td_\star \phi_X f_v)
\]
where the inverse $f^{-1}$ is well defined on $\image(f)$ since we assume $\ker f$ is trivial.
Conversely, if there exists $X \in \grpG$ and $v \in \vecV$ such that $\td_\star \phi_X f_v \not\in \image(f)$, then the action $\psi$ is not well defined and the associated system is not equivariant.

A key contribution of this paper is to show that the velocity space $\vecV$ can be extended to a larger velocity space $\calV$ on which $\td_\star \phi$ is closed and that this defines an extended equivariant system in a natural manner.
We define the \textit{extended velocity space},
\begin{align} \label{eq:extended_space}
    \calV = \text{span} \left\{ \td_\star \phi_X f_v \; \vline \; X \in \grpG, v \in \vecV \right\} \subset \gothX(\calM).
\end{align}
That is, $\calV$ is the smallest vector subspace of $\gothX(\calM)$ containing the image of $\td_\star \phi_X \circ f : \vecV \to \gothX(\calM)$ over all $X$.
Note that $\calV$ may be infinite-dimensional, depending on the structure of $f$ and $\phi$.

Define an \textit{extended (kinematic) system} $\calf: \calV \to \gothX(\calM)$ by the identity inclusion, that is, the linear map that identifies elements $\calf \in \calV$ to elements $\calf \in \gothX(\calM)$.
To perform computations with elements of the input space of the extended system it is convenient to introduce coordinates for $\calV$.
Through the original input function $f$, one may identify the coordinates of the velocity space $\vecV$ with coordinates of a subspace of the extended velocity space, $\vecV \cong f[\vecV] \leq \calV$.
That is, one writes coordinates $u \in \vecV$ for the vector field $f_u \in \calV$.
The coordinates used are an extension of these coordinates to a basis for the full extended velocity $\calV$ space.
Elements of the basis of $\calV$ are written using lower case letters $u, v, w$ and the corresponding vector fields $\calf_u, \calf_v, \calf_w \in \calV \subset \gothX(\calM)$ are denoted with the calligraphic font.
One has that $\calf_u = f_u$ for $u \in \vecV$ since the basis was generated by extension.
For any element $\calf_w \in \calV$ then recalling \eqref{eq:extended_space} there is a finite collection of $\{v_i\} \in \vecV$ and $\{ X_i \} \in \grpG$ such that
\begin{align} \label{eq:extended_kinematic_system}
    \calf_w & = \sum_{i=1}^n \td_\star \phi_{X_i} f_{v_i}.
\end{align}
The calligraphic font $\calf$ and roman font $f$ for functions are difficult to distinguish, however, since the two notations correspond on $\vecV$ and are identical in role, this is an advantage rather than a disadvantage in notation.


\begin{proposition} \label{prop:extended_system_equivariance}
    Consider a kinematic system \eqref{eq:kinematic_system} with symmetry group $\grpG$ and action $\phi: \grpG \times \calM \to \calM$.
    Construct the extended kinematic system as in \eqref{eq:extended_kinematic_system} with system function $\calf : \calV \to \gothX(\calM)$.
    Define $\psi: \grpG \times \calV \to \calV$ as the restriction of $\td_\star \phi$ to the extended velocity space $\calV$.
    That is,
    \[
    \psi(X,w) := \td_\star \phi_X \calf_w,
    \]
    for every $X \in \grpG$ and $w \in \calV$.
    Then $\psi$ is a well defined group action on $\calV$ and the extended system is equivariant with respect to the actions $\psi$ and $\phi$.
\end{proposition}

\begin{proof}
First, it is necessary to show that $\psi$ is well-defined and a group action.
Since $\td_\star \phi$ is a group action, $\psi$ must be a group action, and all that remains to show is that $\psi$ is well-defined.
It suffices to show that $\td_\star \phi_X(w) \in \calV$ for every $w \in \calV$ and $X \in \grpG$.
For any $w \in \calV$ and $Y \in \grpG$ then
\begin{align*}
    \psi(Y, w)
    &= \td_\star \phi_Y \calf_w, \\
    &= \td_\star \phi_Y \sum_{i=1}^n \td_\star \phi_{X_i} f_{v_i}, \\
    &= \sum_{i=1}^n \td_\star \phi_{X_i Y} f_{v_i} \in \calV.
\end{align*}
To show equivariance, consider that, by definition,
\begin{align*}
    \calf_{\psi_X(w)}
    &= \psi_X(w)
    = \td_\star \phi_X \calf_w,
\end{align*}
for every $X \in \grpG$ and $w \in \calV$.
Let $u \in \calV$ be arbitrary, and let $w = \psi_{X}(u)$.
Then,
\begin{align*}
    \calf_{u}(\xi) &= \calf_{\psi_{X^{-1}}(w)}(\xi) \\
    &= \td_\star \phi_{X^{-1}} \calf_w(\xi) \\
    &= \td \phi_{X^{-1}} \calf_w(\phi_{X}(\xi)).
\end{align*}
Thus, applying $\td \phi_{X}$ to both sides,
\[
\td \phi_{X} \calf_{u}(\xi) = \calf_{\psi_X(u)}(\phi_{X}(\xi)),
\]
which proves the extended system is equivariant.
\end{proof}

\subsection{Equivariant Lift}

Suppose the system \eqref{eq:kinematic_system} is equivariant with state symmetry $\phi : \grpG \times \calM \to \calM$ and input symmetry $\psi : \grpG \times V \to V$ for a general input linear vector space $V$.
A lift $\Lambda$ for the system is equivariant if
\begin{align} \label{eq:equivariant_lift}
    \Ad_X \Lambda(\phi(X, \xi), \psi(X, v)) = \Lambda(\xi, v),
\end{align}
for all $X \in \grpG$, $\xi \in \calM$, and $v \in V$.
Equivalently, the diagram
\begin{equation*}
    \xymatrix{
        \gothg \ar[r]^{\Ad_X} & \gothg \\
        \calM \times V \ar[r]^{\phi_X \times \psi_X} \ar[u]^\Lambda & \calM \times V \ar[u]_\Lambda
    }
\end{equation*}
commutes for all $X \in \grpG$.



In this paper, we will assume the existence of an equivariant lift for the system considered and we demonstrate equivariance of the lift chosen for the example given in \S\ref{sec:example}.
In \cite{2020_mahony_annrev}, it is shown that an equivariant lift exists for any equivariant system and an algorithm for construction of this lift is provided for certain cases.






\section{Equivariant Filter}
\label{sec:eqf_design}


The filter design proposed in the sequel requires an equivariant system with an equivariant lift.
A key result from \S\ref{sec:kinematic_equivariance} is that any system with a state symmetry can be extended to an equivariant system and we claim that any such system admits an equivariant lift.
Thus, the proposed filter design can be applied to any system with a homogeneous state symmetry by extending the system if necessary.


\subsection{Observer Architecture}


Consider a kinematic system as in \eqref{eq:kinematic_system} with a state symmetry $\phi: \grpG \times \calM \to \calM$.
Assume the system is equivariant with input symmetry $\psi: \grpG \times \vecV \to \vecV$.
Let $\Lambda : \calM \times \vecV \to \gothg$ be an equivariant lift for this system.
Given a fixed but arbitrary $\mr{\xi} \in \calM$ and a known input signal $v(t) \in \vecV$, the \textit{lifted system} \cite{RM_2013_Mahony_nolcos} is defined by the ODE
\begin{align} \label{eq:lifted_system}
    \dot{X} &= \td L_X \Lambda(\phi(X, \mr{\xi}), v), \quad    \phi(X(0), \mr{\xi}) = \xi(0).
\end{align}
where $X(0) \in \stab_\phi (\mr{\xi})$ is unknown.
The trajectory of the lifted system projects down to the original system trajectory by \cite{RM_2013_Mahony_nolcos}
\begin{align} \label{eq:lifted_system_projection}
    \phi(X(t), \mr{\xi}) \equiv \xi(t).
\end{align}
Define the state observer on the group as $\hat{X} \in \grpG$, with
\begin{align} \label{eq:observer_ode}
    \dot{\hat{X}} &= \td L_{\hat{X}} \Lambda(\phi(\hat{X}, \mr{\xi}), v) + \td R_{\hat{X}} \Delta, \quad  \hat{X}(0) = \id,
\end{align}
where the innovation term $\Delta$ remains to be chosen.

The state estimate of the observer is given by the projection
\[
\hat{\xi} =  \phi(\hat{X}(t), \mr{\xi}).
\]
Thus, the state of the observer is posed on the symmetry group rather than the state-space of the system kinematics.
The fact that the lifted system \eqref{eq:lifted_system} has virtual inputs does not impact on the effectiveness of using it as an internal model for the observer since the input required to implement \eqref{eq:observer_ode} is measured.
The innovation term $\Delta$ will be chosen by applying a Kalman-Bucy filter to linearised invariant error dynamics.

\subsection{Error Dynamics}

Let $\xi \in \calM$ be the true state of the system.
Choose a fixed origin $\mr{\xi} \in \calM$ and let $\hat{X} \in \grpG$ be a state observer with dynamics given by \eqref{eq:observer_ode}.
Define the global state error
\begin{align} \label{eq:state_error_dfn}
    e := \phi(\hat{X}^{-1},\xi).
\end{align}
Note that $\phi(\hat{X},\mr{\xi}) = \xi$ if and only if $e = \mr{\xi}$, and the goal of the filter design will be to drive $e \to \mr{\xi}$.
Define the \emph{origin velocity}
\begin{align} \label{eq:origin_velocity}
\mr{v} := \psi(\hat{X}^{-1}, v).
\end{align}
Note that this is measurable since both $\hat{X}(t)$ and $v(t)$ are available to the observer.

\begin{lemma}
Let the state error $e$ be defined as in \eqref{eq:state_error_dfn} and the origin velocity $\mr{v}$ be defined as in \eqref{eq:origin_velocity}.
The dynamics of $e$ are given by
\begin{align}
    \dot{e} = \td \phi_{e} \left( \Lambda(e, \mr{v}) - \Lambda(\mr{\xi}, \mr{v}) - \Delta \right), \label{eq:state_error_kinematics}
\end{align}
and depend only on $\mr{v}$ and $e$.
\end{lemma}

\begin{proof}
Let $X(t)$ be a solution to the lifted system \eqref{eq:lifted_system} satisfying $\phi(X(0), \mr{\xi}) = \xi(0)$.
Then,
\begin{align} \label{eq:error_state_group_rel}
    e = \phi(\hat{X}^{-1},\xi)
    = \phi(X \hat{X}^{-1}, \mr{\xi})
    = \phi(E, \mr{\xi}),
\end{align}
where $E := X \hat{X}^{-1}$.
Computing the dynamics of $E$, one has
\begin{align} \label{eq:group_error_kinematics}
    \dot{E} &= \td R_{\hat{X}^{-1}} \td L_X \Lambda(\phi(X, \mr{\xi}), v) \notag \\
    &\hspace{0.5cm} - \td L_X \td L_{\hat{X}^{-1}} \td R_{\hat{X}^{-1}} \left( \td L_{\hat{X}} \Lambda(\phi(\hat{X}, \mr{\xi}), v) + \td R_{\hat{X}} \Delta \right), \notag \\
    &= \td L_E \Ad_{\hat{X}} \left( \Lambda(\phi(X, \mr{\xi}), v) - \Lambda(\phi(\hat{X}, \mr{\xi}), v) \right) - \td L_E \Delta, \notag \\
    &= \td L_E \left( \Lambda(\phi(E, \mr{\xi}), \psi(\hat{X}^{-1}, v))
    - \Lambda(\mr{\xi}, \psi(\hat{X}^{-1}, v)) \right)
    \notag \\ &\hspace{2cm}
    - \td L_E \Delta,
\end{align}
where the last line follows from the equivariance of the lift.
Recall the definition $\mr{v} = \psi(\hat{X}^{-1}, v)$.
The dynamics of $e$ follow from \eqref{eq:error_state_group_rel} and \eqref{eq:group_error_kinematics},
\begin{align}
\dot{e}
&= \td \phi_{\mr{\xi}} \dot{E}, \notag \\
&= \td \phi_{\phi(E, \mr{\xi})} \td L_{E^{-1}} \dot{E}, \notag \\
&= \td \phi_{e} \left( \Lambda(\phi(E, \mr{\xi}), \mr{v}) - \Lambda(\mr{\xi}, \mr{v}) - \Delta \right), \notag \\
&= \td \phi_{e} \left( \Lambda(e, \mr{v}) - \Lambda(\mr{\xi}, \mr{v}) - \Delta \right), \notag
\end{align}
as required.
\end{proof}


\subsection{Linearisation}


In contrast to the Extended Kalman Filter (EKF), the linearisation of the state error used in the Equivariant filter does not depend on the estimated state $\hat{\xi}$ of the system, but only on the specific choice of origin $\mr{\xi}$ and the origin velocity $\mr{v}$.

Let $e \in \calM$ denote the state error \eqref{eq:state_error_dfn}, and let $\mr{v} \in \vecV$ denote the origin velocity \eqref{eq:origin_velocity}.
Fix a local coordinate chart $\varepsilon : U \to \R^M$ where $U \subset \calM$ is a neighbourhood of the fixed origin $\mr{\xi}$.
We overload the notation $\varepsilon$ to represent both the chart and the local coordinates of the state error,
\begin{align} \label{eq:linearised_state_error_dfn}
    \varepsilon = \varepsilon(e),
\end{align}
as is common in the literature of smooth manifolds \cite{2009_lee_book_manifolds}.
The EqF is designed around the linearised dynamics of $\varepsilon$.

\begin{proposition} \label{prop:epsilon_dynamics}
    Let $\varepsilon$ be local coordinates for the state error $e$ as in \eqref{eq:linearised_state_error_dfn}.
    The linearised dynamics of $\varepsilon$ about $\varepsilon = 0$ and $\Delta = 0$ are
    \begin{align}
        \dot{\varepsilon} &\approx \mr{A}_t \varepsilon - \td \varepsilon \cdot \tD \phi_{\mr{\xi}} (\id) [\Delta], \label{eq:linearised_state_error_dyn} \\
        \mr{A}_t &:= \td \varepsilon \cdot \tD \phi_{\mr{\xi}} (\id) \cdot \tD_{\mr{\xi}} \Lambda(\mr{\xi}, \mr{v}) \cdot \td \varepsilon^{-1} \label{eq:state_matrix_A_dfn},
    \end{align}
    which depend only on the origin velocity $\mr{v}$ \eqref{eq:origin_velocity} as an external input.
\end{proposition}

\begin{proof}
Consider the nonlinear dynamics of the global state error \eqref{eq:state_error_kinematics}
\begin{align}
    \dot{e} & =  \td \phi_{e} \left( \Lambda(e, \mr{v}) - \Lambda(\mr{\xi}, \mr{v})\right)  - \td \phi_{e} \Delta, \notag \\
    & = \td \phi_e \tilde{\Lambda}_{\mr{\xi}} (e,\mr{v}) - \td \phi_{e} \Delta, \notag \\
    \tilde{\Lambda}_{\mr{\xi}}(e, \mr{v}) &:=  \Lambda(e, \mr{v}) - \Lambda(\mr{\xi}, \mr{v})
\label{eq:alpha}
\end{align}
where $\td \phi_e \tilde{\Lambda}_{\mr{\xi}}(e, \mr{v})$ models the error drift term when the innovation $\Delta = 0$.
Then,
\begin{align} \label{eq:vareps_nldyn}
    \dot{\varepsilon} &= \td \varepsilon \cdot \td \phi_e \tilde{\Lambda}_{\mr{\xi}}(\varepsilon^{-1}(\varepsilon), \mr{v}) - \td \varepsilon \cdot \td \phi_{e} \Delta.
\end{align}
Clearly, $\tilde{\Lambda}_{\mr{\xi}}(\mr{\xi}, \mr{v}) \equiv 0$.
Hence, linearising $\td \varepsilon \cdot \td \phi_e \tilde{\Lambda}_{\mr{\xi}}(\varepsilon^{-1}(\varepsilon), \mr{v})$ about $\varepsilon = 0$ yields
\begin{align*}
    \td \varepsilon \cdot &\td \phi_e \tilde{\Lambda}_{\mr{\xi}}(\varepsilon^{-1}(\varepsilon), \mr{v}) \\
    &\approx \td \varepsilon \cdot \td \phi_{\mr{\xi}} \cdot \tilde{\Lambda}_{\mr{\xi}}(\varepsilon^{-1}(0), \mr{v}) \\
    &\hspace{1cm} + \td \varepsilon \cdot \td \phi_{\mr{\xi}} \cdot \tD_{\mr{\xi}} \tilde{\Lambda}_{\mr{\xi}}(\varepsilon(^{-1}(0)), \mr{v}) \cdot \td \varepsilon^{-1} \varepsilon, \\
    &= \td \varepsilon \cdot \td \phi_{\mr{\xi}} \cdot \tD_{\mr{\xi}} \tilde{\Lambda}_{\mr{\xi}} (\mr{\xi}, \mr{v}) \cdot \td \varepsilon^{-1} \varepsilon, \\
    &= \td \varepsilon \cdot \tD \phi_{\mr{\xi}} (\id) \cdot \tD_{\mr{\xi}} \Lambda(\mr{\xi}, \mr{v}) \cdot \td \varepsilon^{-1} \varepsilon, \\
    &= \mr{A}_t \varepsilon.
\end{align*}
This gives the linearisation of the first term in \eqref{eq:vareps_nldyn}
Since we are linearising about $\Delta = 0$, the linearisation of $\dot{\varepsilon}$ \eqref{eq:vareps_nldyn} is
\begin{align*}
    \dot{\varepsilon} &\approx \mr{A}_t \varepsilon - \td \varepsilon \cdot\tD \phi_{\mr{\xi}}(\id) [\Delta],
\end{align*}
as required.
\end{proof}

Observe that the linearised dynamics of $\varepsilon$ depend on $\tD \phi_{\mr{\xi}} (\id) \Delta$, that is, they depends linearly on $\Delta$.
In the case where the group $\grpG$ is of higher dimension than the state space $\calM$, the map $\tD \phi_{\mr{\xi}} (\id)$ is not injective, and the component of $\Delta$ in the kernel is absent in the linearised dynamics of $\varepsilon$.

Consider the state error $e$ \eqref{eq:state_error_dfn}, and let $\xi \in \calM$ and $\hat{X} \in \grpG$ denote the true system state and observer state, respectively.
Let $\varepsilon \in \R^M$ represent local coordinates for $e$ as in \eqref{eq:linearised_state_error_dfn}.
The system output $h(\xi) \in \calN$ can be written
\begin{align} \label{eq:h_in_error_coords}
    h(\xi) = h(\phi(\hat{X}, e)) = h(\phi_{\hat{X}}(\varepsilon^{-1}(\varepsilon))).
\end{align}
Note that substituting $\varepsilon = 0$ gives
\begin{align*}
h(\phi_{\hat{X}}(\varepsilon^{-1}(0))) & =
h(\phi_{\hat{X}}(\mr{\xi}))
= h(\hat{\xi}).
\end{align*}
The measurement residual $h(\xi) - h(\hat{\xi}) \in \R^N$ is computed using the embedding of $\calN$ into $\R^N$.
Linearising $h(\xi) - h(\hat{\xi})$ in $\R^N$ as a function of $\varepsilon \in \R^M$ around the point $\varepsilon = 0$ yields
\begin{align}
h(\xi) - h(\hat{\xi}) & \approx C_t \varepsilon, \label{eq:linearised_measurement_error_dyn} \\
C_t &:= \tD h (\hat{\xi}) \cdot \tD \phi_{\hat{X}} (\mr{\xi}) \cdot \td \varepsilon^{-1}, \label{eq:output_matrix_C_dfn}
\end{align}
where $C_t \in \R^{N \times M}$ is the differential of \eqref{eq:h_in_error_coords} at the point $\varepsilon = 0$.




\subsection{Equivariant Filter (EqF)}
\label{sec:eqf_definition}


Consider a kinematic system as in \eqref{eq:kinematic_system} with a state symmetry $\phi: \grpG \times \calM \to \calM$.
Assume the system is equivariant with input symmetry $\psi: \grpG \times \vecV \to \vecV$ and has an equivariant lift $\Lambda : \calM \times \vecV \to \gothg$.
Let $\xi \in \calM$ denote the true state of the system, with trajectory determined by the measured input $v(t) \in \vecV$.

We construct the Equivariant Filter (EqF) as follows.
Let $\hat{X} \in \grpG$ denote the observer state.
Pick an arbitrary fixed origin $\mr{\xi} \in \calM$, and let $\mr{A}_t$ and $C_t$ be defined as in \eqref{eq:state_matrix_A_dfn} and \eqref{eq:output_matrix_C_dfn} by choosing a local coordinate chart $\varepsilon$ around $\mr{\xi}$.
Choose an initial value for the Riccati term $\Sigma_0 \in \Sym_+(M)$, where $\Sym_+(M)$ is the set of positive-definite symmetric $M \times M$ matrices, and pick a state gain matrix $P \in \Sym_+(M)$ and an output gain matrix $Q \in \Sym_+(N)$.
Choose a right inverse $\tD \phi_{\mr{\xi}}(\id)^\dag$ of $\tD \phi_{\mr{\xi}}(\id)$; that is, $\tD \phi_{\mr{\xi}}(\id) \cdot \tD \phi_{\mr{\xi}}(\id)^\dag = \id$.

Define the EqF dynamics to be
\begin{align}
    \dot{\hat{X}} &= \td L_{\hat{X}} \Lambda(\phi(\hat{X}, \mr{\xi}), v) + \td R_{\hat{X}} \Delta, \quad \hat{X}(0) = \id \label{eq:eqf_group_observer} \\
    \Delta &= \tD \phi_{\mr{\xi}}(\id)^\dag \td \varepsilon^{-1} \Sigma C_t^\top Q^{-1} (h(\xi) - h(\hat{\xi})), \label{eq:eqf_delta_innovation} \\
    \dot{\Sigma} &= \mr{A}_t \Sigma + \Sigma {\mr{A}_t}^\top + P - \Sigma C_t^\top Q^{-1} C_t \Sigma, \quad \Sigma(0) = \Sigma_0\label{eq:eqf_riccati}
\end{align}
where the subtraction $h(\xi) - h(\hat{\xi})$ \eqref{eq:eqf_delta_innovation} is understood by using the embedding of $\calN \hookrightarrow \R^N$.

If the pair $(\mr{A}_t, C_t)$ is uniformly observable, then $\Sigma(t)$ is bounded above and below, and hence the Riccati equation \eqref{eq:eqf_riccati} and innovation $\Delta$ \eqref{eq:eqf_delta_innovation} are well-defined for all time \cite{2001_deylon_tac}.
A key property of the EqF is that $\mr{A}_t$ is independent of the state, as it arises from the linearisation of the error state around a fixed origin $\mr{\xi}$, and depends only on $\mr{v}$.


\section{Example: 2D derotated ego-centric visual SLAM}
\label{sec:example}


In this section we consider the problem of 2D derotated ego-centric visual SLAM.
Note that while the proposed EqF does provide an interesting and novel approach to SLAM problem considered, the focus of the present work is to demonstrate the design procedure rather than to benchmark the performance of the resulting filter against previously published algorithms.
The present example has been chosen to be as simple algebraically as possible while capturing the particular properties of the SLAM problem that make the equivariant filter approach of interest.

Consider a robot moving on a 2D plane with a camera that is observing the bearings of 2D landmarks in the plane.
Assume that the camera has a full $360^\circ$ field of view (ensuring landmarks do not leave the field of view of the camera) and the image has been derotated so that it has a fixed orientation in the reference frame (allowing the angular velocity of the robot to be ignored).
Our objective is to estimate the coordinates of a collection of observed landmarks in the derotated \emph{body-fixed frame} of a moving robot.
That is, we estimate the ego-centric structure of the environment.
We do not try to estimate the evolution of the robot pose or consider the rotation of the robot.


\subsection{System Definition}


Let $v \in \R^2$ be the inertial velocity of the robot and $x_i \in \R^2$ for $i=1,..,n$ be the coordinates of $n$ stationary points expressed in the \emph{body-fixed frame}.
The camera measures the bearing to each point, $y_i := \frac{x_i}{\Vert x_i \Vert}$.
The body-fixed origin of the robot is chosen at the focal point of the camera.
The ego-centric kinematics of the points in the body-fixed frame are
\begin{align}
    \dot{x}_i &= f_v(x_i) := -v, \label{eq:bvslam_kinematics}\\
    y_i &= h(x_i) := \frac{x_i}{\Vert x_i \Vert} \label{eq:bvslam_measurement}.
\end{align}
In this system, the velocity space is $\vecV := \R^2$, the state space is $\calM := (\R^{2} \setminus \{0\})^n$, and the measurement space is $\calN := (\Sph^1)^n$.
Note that $x_i = 0$ must be excluded from the state space as the measurement function is not defined at these points.

Define the product Lie group $\grpG := ({\bf \Sph^1} \times \MR)^n$, where ${\bf \Sph^1}$ is the 1-sphere under addition, and $\MR$ is the positive reals under multiplication.
Define $\phi:\grpG \times \calM \to \calM$ by
\begin{align} \label{eq:bvslam_symmetry_phi}
    \phi((\theta_i, a_i), x_i) := a_i^{-1} R(\theta_i)^\top x_i,
\end{align}
where $R(\theta_i) \in \SO(2)$ is the rotation matrix constructed from the angle $\theta_i \in \Sph^1$.
Then, $\phi$ is a clearly a smooth, transitive right group action of $\grpG$ on $\calM$.


To see that this system is not equivariant \eqref{eq:equivalent_def_equivariance} consider the following computation.
Given $(\theta_i, a_i) \in \grpG$ and $v \in \R^2$,
\begin{align}
    \td_\star \phi_{(\theta_i, a_i)} f_v(x_i)
    &= \td \phi_{(\theta_i, a_i)} f_v(\phi_{(\theta_i, a_i)^{-1}}(x_i)), \notag \\
    &= \tD \phi_{(\theta_i, a_i)} (a_i R(\theta_i) x_i) f_v(a_i R(\theta_i) x_i), \notag \\
    &= \tD \phi_{(\theta_i, a_i)} (a_i R(\theta_i) x_i)[-v], \notag \\
    &= -a_i^{-1} R(\theta_i)^\top v. \label{eq:example_vector_space_action}
\end{align}
For the system to be equivariant, this holds for each $i = 1, \ldots, n$, and for any choices of $(\theta_i, a_i) \in \grpG$ and $v \in \R^2$.
In particular, if one has $f_w = \td_\star \phi_{(\theta_i, a_i)} f_v(x_i)$ then
\begin{align}
w = -a_i^{-1} R(\theta_i)^\top v
\label{eq:w_not_defined}
\end{align}
for $i = 1, \ldots, n$.
For general $a_i \not= a_j$ and $\theta_i \not= \theta_j$ then \eqref{eq:w_not_defined} does not hold and $\td_\star \phi_{(\theta_i, a_i)} f_v(x_i)$ does not lie in $\image(f)$.
To implement the EqF, this system requires a velocity extension.


\subsection{Velocity Extension and Kinematic Equivariance}


Let $(v_i) = (v_1,...,v_n) \in \R^{2n} \equiv V$ and define
$\calf_{(v_i)}(x_i) = v_i$ to be the vector field associated with extended kinematics
\begin{align}
\dot{x}_i &= \calf_{(v_i)} (x_i) := -v_i.  \label{eq:ext_bvslam_kinematics}
\end{align}
Let $\calV = \image(\calf) : V \to \gothX(\calM)$ to be the associated subspace of vector fields.
Thus, $V \equiv \R^{2n}$ will play the role of coordinates for the extended input space and $\calV$ denotes the associated vector fields.
Note that the original kinematics \eqref{eq:bvslam_kinematics} are contained in the extended system by setting $(v_i) = (v, \ldots, v)$.

Given any $v \in \vecV$ and $(\theta_i, a_i) \in \grpG$ then define $\psi : \grpG \times V \to V$ by
\begin{align}
 \psi((\theta_i, a_i),v_i) :=  a_i^{-1} R(\theta_i)^\top v_i.
\label{eq:psi_example}
\end{align}
It is easily verified that $\psi$ is a group action on $V$.

One has that
\begin{align}
\td_\star \phi_{(\theta_i, a_i)} f_{(v_i)}(x_i) &= -a_i^{-1} R(\theta_i)^\top v_i \label{eq:inline_1}\\
& =  - \psi((\theta_i, a_i),v_i) \label{eq:inline_2} \\
& = \calf_{\psi((\theta_i, a_i),v_i)}(x_i).\label{eq:inline_3}
\end{align}
where \eqref{eq:inline_1} follows from \eqref{eq:example_vector_space_action}, and \eqref{eq:inline_2} and \eqref{eq:inline_3} follow from the definitions of $\psi$ and $\calf$, respectively.
Thus $\td_\star \phi : \image (f) \to \calV$, and consequently, $\calV$ is closed under $\td_\star \phi$.

Therefore, $\calf$ is a kinematic extensions of $f$, and the extended kinematic system \eqref{eq:ext_bvslam_kinematics} is equivariant since
\begin{align*}
\td \phi_{(\theta_i, a_i)} \calf_{(v_i)}(x_i) & = \td_\star \phi_{(\theta_i, a_i)} \calf_{(v_i)}(\phi_{(\theta_i, a_i)}(x_i)) \\
& = \calf_{\psi_{(\theta_i,a_i)}(v_i)}(\phi_{(\theta_i, a_i)}(x_i)).
\end{align*}

\subsection{Equivariant Lift}


Define the function $\Lambda: \calM \times \calV \to \gothg$ as
\begin{align}
    \Lambda(x_i, v_i) &= \left( \frac{x_i^\top S v_i}{x_i^\top x_i}, -\frac{x_i^\top v_i}{x_i^\top x_i} \right), \\
    S &= \begin{pmatrix}
        0 & -1 \\ 1 & 0
    \end{pmatrix},
\end{align}
where $S$ has the properties $S^\top = -S$ and $z^\top S z = 0$ for any $z \in \R^2$.
Then $\Lambda$ is an equivariant lift for the extended kinematic system \eqref{eq:ext_bvslam_kinematics} with symmetry $\phi$ \eqref{eq:psi_example}.

To verify that $\Lambda$ is a lift \eqref{eq:lift_condition}:
\begin{align*}
    \tD \phi_{x_i} (\id) \left[\Lambda(x_i, v_i)\right]
    &= \tD \phi_{x_i} ((0, 1)) \left[ \frac{x_i^\top S v_i}{x_i^\top x_i}, -\frac{x_i^\top v_i}{x_i^\top x_i} \right], \\
    &= \frac{x_i^\top v_i}{x_i^\top x_i} x_i + \frac{x_i^\top S v_i}{x_i^\top x_i} S^\top x_i, \\
    &= \frac{x_i x_i^\top}{x_i^\top x_i} v_i + \frac{(S^\top x_i)(S^\top x_i)^\top}{x_i^\top x_i} v_i, \\
    &= v_i,
\end{align*}
where the last line follows from $\frac{(S^\top x_i)(S^\top x_i)^\top}{x_i^\top x_i} = I_2 - \frac{x_i x_i^\top}{x_i^\top x_i}$.


To verify that $\Lambda$ is equivariant \eqref{eq:equivariant_lift}:
\begin{align}
    \Ad & _{(\theta_i, a_i)} \Lambda(\phi_{(\theta_i, a_i)}(x_i), \psi_{(\theta_i, a_i)}(v_i)) \notag \\
    \label{eq:step:abelian_adjoint}
    &= \Lambda ( a_i^{-1} R(\theta_i)^\top x_i, a_i^{-1} R(\theta_i)^\top v_i ), \\
    &= \left( \frac{(a_i^{-1} R(\theta_i)^\top x_i)^\top S (a_i^{-1} R(\theta_i)^\top v_i)}{\Vert a_i^{-1} R(\theta_i)^\top x_i \Vert^2},
    \right. \notag \\ &\hspace{2cm} \left.
    - \frac{(a_i^{-1} R(\theta_i)^\top x_i)^\top a_i^{-1} R(\theta_i)^\top v_i}{\Vert a_i^{-1} R(\theta_i)^\top x_i \Vert^2} \right), \notag \\
    &= \left( \frac{x_i^\top R(\theta_i) S R(\theta_i)^\top v_i}{\Vert x_i \Vert^2}, - \frac{ x_i^\top R(\theta_i) R(\theta_i)^\top v_i}{\Vert x_i \Vert^2} \right), \notag \\
    &= \left( \frac{x_i^\top S v_i}{\Vert x_i \Vert^2}, - \frac{ x_i^\top v_i}{\Vert x_i \Vert^2} \right), \notag \\
    &= \Lambda(x_i, v_i), \notag
\end{align}
where \eqref{eq:step:abelian_adjoint} follows from the fact that $\grpG$ is abelian.

\subsection{EqF Implementation}


We follow the EqF design procedure as described in \S\ref{sec:eqf_definition}.

We use the standard embedding of $\calN = (\Sph^1)^n \hookrightarrow (\R^2)^n$, that is, given $y \in \Sph^1$, we identify $y$ with its coordinates as a unit vector in $\R^2$.
Finally, we identify the tangent space $\tT_{x_i} \calM$ with $\R^{2n}$ for every $(x_i) \in \calM = (\R^{2} \setminus \{0\})^n$.

Fix an origin element $\mr{x}_i \in \calM$, and let $(\hat{\theta}_i, \hat{a}_i) \in \grpG$ denote the observer state.
The linearised error state matrix $\mr{A}_t$ and linearised error output matrix $C_t$ for this system can be explicitly written in terms of the origin element $\mr{x}_i$ and the origin velocity $\mr{v} := \psi((\hat{\theta}_i, \hat{a}_i)^{-1}, v)$ by choosing to identify the neighbourhood around $(\mr{x}_i)$ with its neighbourhood in $\R^{2n}$.
\begin{align*}
    \mr{A}_t &= \text{diag}(A_i), \\
    A_i &= \frac{1}{\Vert \mr{x}_i \Vert^4} \left( \Vert \mr{x}_i \Vert^2\left( - (S\mr{x}_i) (S\mr{v}_i)^\top + \mr{x}_i {\mr{v}_i}^\top  \right)
    \right. \\ &\hspace{2cm} \left.
    + 2 {\mr{x}_i}^\top S \mr{v}_i (S \mr{x}_i) {\mr{x}_i}^\top - 2{\mr{x}_i}^\top \mr{v}_i {\mr{x}_i} {\mr{x}_i}^\top \right), \\
    C &= \text{diag}(C_i), \\
    C_i &= R(\theta_i)^\top \frac{1}{\Vert \mr{x}_i \Vert }\left( I_2 - \frac{\mr{x}_i {\mr{x}_i}^\top}{{\mr{x}_i}^\top \mr{x}_i} \right),
\end{align*}
In this system, $\tD \phi_{\mr{x}_i}(\id)$ is full rank and thus has a unique inverse,
\begin{align*}
    (\tD \phi_{\mr{x}_i}(\id)^{-1}) = \diag \begin{pmatrix}
        {\mr{x}_i}^\top S \\ - {\mr{x}_i}^\top
    \end{pmatrix},
\end{align*}
which is required to compute the innovation $\Delta \in \gothg$ as in \eqref{eq:eqf_delta_innovation}.

\begin{lemma} \label{lem:example_observability}
    Assume that $\Vert v \Vert$ and $\Vert x_i \Vert$ are bounded, and that there exist $\tau, \delta > 0$ such that
    \begin{align*}
    \frac{1}{\tau} \int_t^{t+\tau} (v(s)^\top S x_i (s))^2 \td s \geq \delta,
    \end{align*}
    for all $i$ and all $t \geq 0$.
    Then the linearised error system is uniformly observable.
\end{lemma}

\begin{proof}
The proof is based on Theorem 3.1 of \cite{2018_Hamel_TAC}.
It has been omitted from this paper to save space.
\end{proof}

Assuming the conditions of Lemma \ref{lem:example_observability} are met, let $\Sigma \in \Sym_+(2n)$ be the solution to the Riccati equation \eqref{eq:eqf_riccati} with state gain matrix $P \in \Sym_+(2n)$ and output gain matrix $Q \in \Sym_+(2n)$.
In this example, we may choose $\Sigma$ to have block diagonal structure $\Sigma = \diag(\Sigma_i)$, $i=1,...,n$.
This requires $P$ and $Q$ also have block diagonal structure.

The EqF observer equations \eqref{eq:eqf_group_observer} specialises to
\begin{align*}
    \ddt (\hat{\theta}_i, \hat{a}_i) &= \td L_{(\hat{\theta}_i, \hat{a}_i)} \Lambda(\phi_{(\hat{\theta}_i, \hat{a}_i)}(\mr{x}_i), v_i) - \td R_{(\hat{\theta}_i, \hat{a}_i)} \Delta, \\
    \dot{\hat{\theta}}_i &= \frac{\hat{x}_i^\top S v_i}{\hat{x}_i^\top \hat{x}_i} - \Delta_{\hat{\theta}_i}, \\
    \dot{\hat{a}}_i &= - \hat{a}_i \frac{\hat{x}_i^\top v_i}{\hat{x}_i^\top \hat{x}_i} - \hat{a}_i \Delta_{\hat{a}_i},
\end{align*}
where the innovation $\Delta$ comes from the specialisation of \eqref{eq:eqf_delta_innovation},
\begin{align*}
    \Delta &= -\tD \phi_{\mr{x}_i}(\id)^{-1} \tD\varepsilon^{-1}(0) \Sigma C_t^\top Q^{-1} (h(x_i) - h(\hat{x}_i)), \\
    &= \diag(\Delta_i), \\
    \Delta_i &= -\begin{pmatrix}
        {\mr{x}_i}^\top S \\ - {\mr{x}_i}^\top
    \end{pmatrix}
    \Sigma_i
    \left( I_2 - \frac{\mr{x}_i {\mr{x}_i}^\top}{{\mr{x}_i}^\top \mr{x}_i} \right) \frac{R(\theta_i)}{\Vert \mr{x}_i \Vert }
    Q^{-1}_i (y_i - \hat{y}_i),
\end{align*}
where $y_i$ and $\hat{y}_i$ denote the $i$th measured and estimated measurements, respectively.




\subsection{Simulation Results}


To verify the observer design for this example, we performed a simulation of an omni-directional camera moving in $\R^2$ with a velocity given by $v = (2\cos(2t), 0)$, where $t$ is the simulation time in seconds, observing 4 stationary landmarks.
The landmarks' initial positions relative to the camera were drawn from a uniform random distribution within $[-0.5, 0.5] \times [1.0, 2.0]$.
All landmarks are observed for all time.

The origin coordinates $\mr{x}_i$ of the observer's landmark estimates were chosen by adding a random offset drawn from a uniform distribution within $[-1,-1]\times[1,1]$ to the true landmark positions.
The observer gain matrices were given a block structure so they could be decoupled, $P = \diag(P_i), Q = \diag(Q_i)$.
The values for the state and output gain matrices were set to $P_i = 0.02^2 I_2$ and $Q_i = 0.01^2 I_2$, respectively.
The initial value of the Riccati matrices were set to $\Sigma_i = 4^2 I_2$.
The observer equations were implemented with a time step of $dt = 0.01$ using Euler integration.

The results of the simulation are presented as Lyapunov function evolutions in Figure \ref{fig:sim_lyapunov}.
For each landmark $i$, the Lyapunov value is defined by
\begin{align*}
    l_i := (e_i-\mr{x}_i)^\top \Sigma_i^{-1} (e_i-\mr{x}_i),
\end{align*}
where $(e_i) = \phi_{(\hat{\theta}_i, \hat{a}_i)^{-1}}(x_i)$ are the state errors, and $\Sigma_i$ are the Riccati solution blocks corresponding to each individual landmark.

\begin{figure}[!htb]
    \centering
    \includegraphics[width=0.75\linewidth]{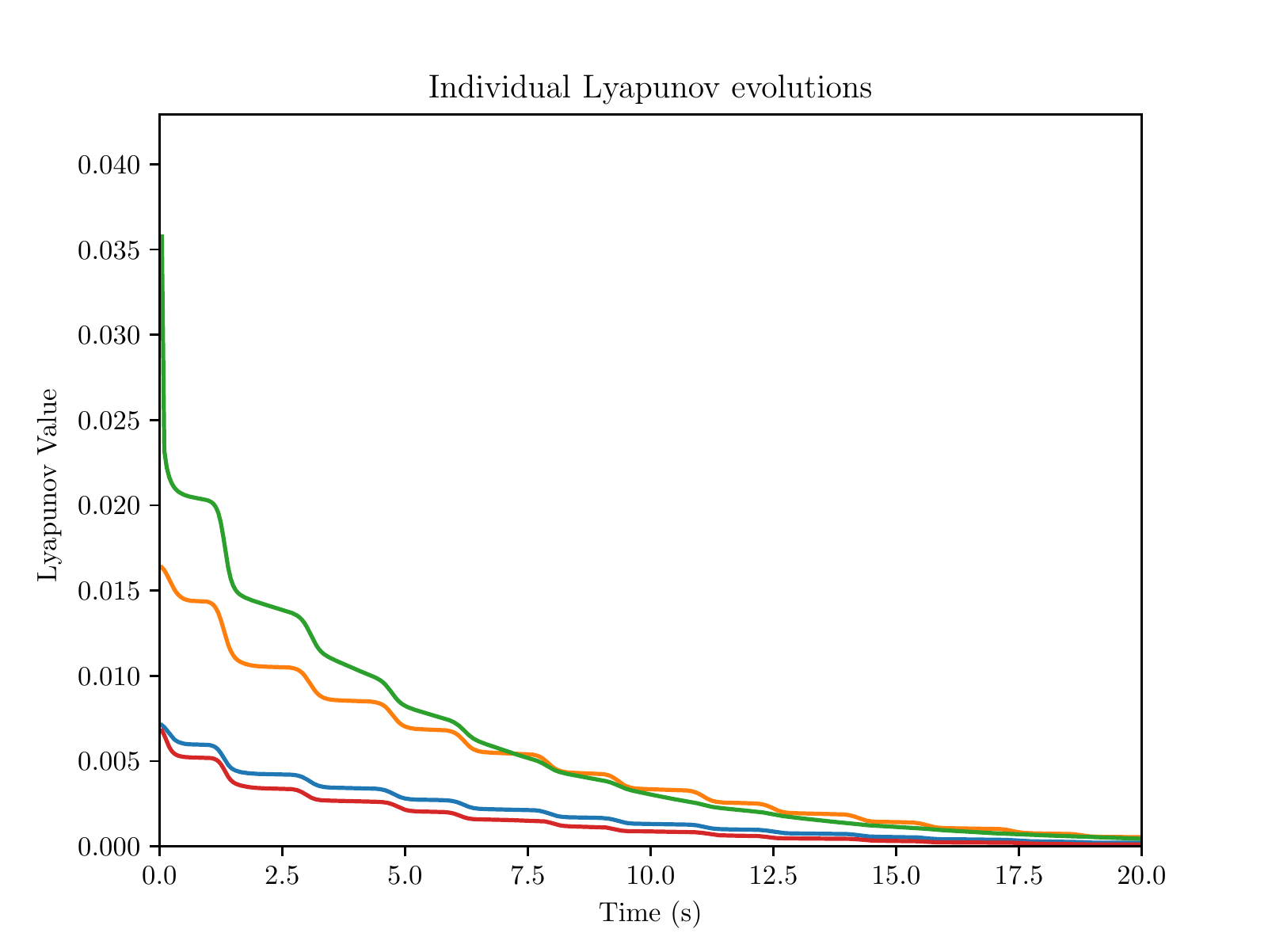}
    \caption{The value of the Lyapunov function for the linearised error system, for four different points.}
    \label{fig:sim_lyapunov}
\end{figure}

The Figure \ref{fig:sim_lyapunov} shows that the local Lyapunov functions converge for all four landmarks.
Observe that, in Figure \ref{fig:sim_lyapunov}, the convergence of the Lyapunov functions speeds up and slows down according to a pattern that repeats every $\pi/2$ seconds.
This is due to the observability of the system, which scales with the square of the velocity as suggested by Lemma \ref{lem:example_observability}, and therefore has a period of $\pi$ seconds.

\subsection{Discussion}

The example system shown is of interest for several reasons.
First, the system is defined on a homogeneous space of a Lie group rather than on a Lie group itself, necessitating the development of a lifted system to apply equivariant observer design methods.
Second, since the original system failed to be equivariant it required a velocity extension to become equivariant.
The existence of such an extension and its construction are novel contributions of this paper.
Finally, the lifted system was shown to be equivariant, but it is not invariant nor group affine, precluding the use of existing observer design methodologies \cite{RM_2013_Mahony_nolcos, 2017_Barrau_tac}.
To the authors' understanding the only general equivariant filter design that can be applied in such a case as this is the proposed EqF.



\section{Conclusion}

Equivariant systems are those with compatible Lie group symmetries of their input and state spaces.
A key contribution of this paper is to show that any system with a state symmetry can be extended to an equivariant system with a compatible input symmetry.
The Equivariant Filter (EqF) has been proposed as a general observer design for equivariant systems, and can be applied to a broader class of systems than designs considered in previous work.
The example of 2D derotated ego-centric visual SLAM is used to detail the system extension procedure and the implementation of the EqF.
A simulation of the EqF for this system demonstrates both the convergence of the EqF estimates as well as the relationship between the observability of the linearised EqF system and the rate of decrease in the observer error.

\section*{Acknowledgment}

This research was supported by the Australian Research Council
through the ``Australian Centre of Excellence for Robotic Vision'' CE140100016.

\bibliographystyle{plain}
\bibliography{references}

\end{document}